\documentclass[journal]{IEEEtran}
\IEEEoverridecommandlockouts
\usepackage{cite}
\usepackage{amsmath,amssymb,amsfonts,amsthm}
\usepackage{algorithmic}

\usepackage{graphicx}
\usepackage{textcomp}
\usepackage{xcolor}
\usepackage{bm}
\usepackage{booktabs}
\def\BibTeX{{\rm B\kern-.05em{\sc i\kern-.025em b}\kern-.08em
	T\kern-.1667em\lower.7ex\hbox{E}\kern-.125emX}}
\newcommand\jz[1]{\begin{pmatrix} #1 \end{pmatrix}}

\newtheorem{theorem}{Theorem}

\newtheorem{remark}{Remark}
\newtheorem{assumption}{Assumption}
\newtheorem{definition}{Definition}

\newtheorem{prop}{Proposition}

\makeatletter
\def\subsubsection{%
	\@startsection
	{subsubsection}                 
	{3}                             
	{\parindent}                    
	{3.5ex plus 1.5ex minus 1.5ex}  
	{0.7ex plus .5ex minus 0ex}     
	{\normalfont\normalsize\itshape}
}
\makeatother

\begin{document}

\title{ Identifying the Most Influential Driver Nodes for Pinning Control of Multi-Agent Systems with Time-Varying Topology\\
	\thanks{This work has been submitted to the IEEE for possible publication. Copyright may be transferred without notice, after which this version may no longer be accessible.\\
		G. Zhang, Z. Liu, X. Yu and M. Jalili are with the School of Engineering, RMIT University, Melbourne, Australia. (Email: guangrui.zhang@hotmail.com; S3935750@student.rmit.edu.au; xinghuo.yu@rmit.edu.au; mahdi.jalili@rmit.edu.au;).}
}

\author{\IEEEauthorblockN{Guangrui Zhang, Zhaohui Liu, Xinghuo Yu and Mahdi Jalili}
}

\maketitle

\begin{abstract}
Identifying the most influential driver nodes to guarantee the fastest synchronization speed is a key topic in pinning control of multi-agent systems. This paper develops a methodology to find the most influential pinning nodes under time-varying topologies. First, we provide the pinning control synchronization conditions of multi-agent systems. Second, a method is proposed to identify the best driver nodes that can guarantee the fastest synchronization speed under periodically switched systems. We show that the determination of the best driver nodes is independent of the system matrix under certain conditions. Finally, we develop a method to estimate the switching frequency threshold that can make the selected best driver nodes remain the same as the average system. Numerical simulations reveal the feasibility of these methods.
\end{abstract}


\section{Introduction}
Multi-agent system (MAS) framework has been frequently used to model real-world systems, where nodes represent the entities in systems and edges describe the connection structures between them \cite{herrera2020multi}. Synchronization phenomenon \cite{9272653}, one of the most intriguing collective behaviours that describe how individuals behave in unison, has attracted much attention. The pinning control approach is proposed to control only a fraction of selected nodes to synchronize the whole system to reference state \cite{9524326}. The selected pinning nodes are denoted as driver nodes. Finding the most influential driver nodes is a key challenge that has received attention in the literature \cite{qiu2021identifying, li2022evaluation}. The most intuitive method is selecting the nodes with the highest centrality as driver nodes, which can be evaluated by their degree or betweenness centrality \cite{qiu2021identifying}. Though this method is easy and efficient, it can not guarantee the optimal performance of pinning control \cite{jalili2015optimal, zhou2015enhancing}. Several optimization algorithms have been proposed to find the most influential driver nodes \cite{orouskhani2016optimizing}. However, these algorithms are often computationally expensive especially for large networks. Recently, based on master stability theory, a new method has been proposed to identify the most influential nodes \cite{amani2017finding, amani2018controllability}. By calculating the eigenratio of the augmented Laplacian matrix of the graph, it can find the most influential driver node to improve the pinning control performance with better controllability \cite{zhou2020identifying}. 

Most of the existing pinning control research mentioned above mainly focus on static models with fixed network topological structure. However, practical engineering systems are more complex in reality than in theory \cite{mahela2020comprehensive, drew2021multi, koohi2020review}. Changes in system structure bring challenges to traditional research approaches with fixed topologies. Hence, new analytical methods need to be proposed to solve the problems in more realistic system models, i.e. MASs with time-varying topologies. Although some research works have been recently published to discuss the MASs with switching typologies \cite{chen2020synchronization, zhang2022pinning, 9140034}, finding the best driver nodes in MASs with time-varying topology still remains an unsolved problem.

This paper analyzes the pinning control problems with time-varying topologies. Firstly, different from the existing methods, the synchronization conditions are given under time-varying topologies without using Lyapunov function. To simplify the model, the switching pattern is considered as periodical switching, and then a new method is proposed to find the most influential pinning nodes to guarantee the fastest synchronization speed converging to the desired states. The results indicate that the choice of the most influential driver node is affected by the switching frequency. Hence, we estimate the switching frequency threshold to guarantee that the most influential driver nodes remain the same under average system. We also prove that the determination of the most influential node is independent of the system matrix under certain conditions. The corresponding simulations are provided to illustrate these theories in the following section.

\section{Preliminaries}
For a matrix $M$, $\lambda_{\max}(M)$ denotes the largest real part of its eigenvalues, and $\lambda_{\max}(M)$ denotes the smallest one. $\rho(M) = |\lambda_{\max}(M)|$ denotes the spectral radius of matrix $M$. $\coprod$ is the consecutive left matrix multiplication.

Consider a MAS consisting of $N$ nodes. The directed graph of it is denoted by $ \mathcal{G} $. The topological structure of MAS can be described by the Laplacian matrix $L = \left [ l_{ij} \right ] \in \mathbb{R}^{N\times N}$. 

In this paper, the topology is assumed to be periodically switched with switching signal $\xi(t):[0,+\infty)\rightarrow \left \{ 1, ..., p \right \}$. The switches occur at $t_k \triangleq k\tau$, $\tau>0$, $k\in\mathbb{N}$. Let $ G_{\xi(t)}$ be the time-varying topology of the network and it keeps invariant in each time interval $[t_{k}, t_{k+1})$. The periodical switching topology means there is a switching period $T$ with $G_{\xi(t+T)} = G_{\xi(t)}$ and $T = m\tau$, $m>0$. Moreover, let $L_{\xi(t)}$ be the corresponding Laplacian matrix of $G_{\xi(t)}$.

The average Laplacian matrix keeps constant because of the periodicity of switching signal: $L_{av} = \frac{1}{T}\int_{t}^{t+T} L_{\xi(t)}~dt$.

\begin{assumption}\label{ass:Laplace}
	The topology induced from the average Laplacian matrix contains a spanning tree.
\end{assumption}

\section{Main Results}
\subsection{Problem Formulation}
For a $N$-nodes MAS with continuous-time states, we denote the state of $i$th node by $x_i (t)\in \mathbb{R}^{n}$, $i=1,2,...,N$. The dynamics equation of it can be written as:
\begin{equation}\label{eq1}
	\dot{x}_{i}(t)=Ax_i(t)-r\sum_{j=1,j\neq i}^{N}l_{ij}(t)\Lambda\left ( x_j(t)-x_i(t) \right) + u_i,
\end{equation}
where $A$ is system matrix, $r$ is coupling strength, $\Lambda \in \mathbb{R}^{n \times n}$ is inner coupling matrix and $u_i$ is probable input of $i$th node.

The first aim is to design a proper controller such that the state of each node can approach a desirable synchronization trajectory $c(t)$ when $t\rightarrow \infty $, i.e.
$\lim_{t\rightarrow \infty }\left \| x_i(t)-c(t) \right \|=0,\quad i=1,2,...,N$,
where $\dot{c}(t)=Ac(t)$. The second aim is to identify the most influential pinning node based on the given conditions. The determination of the most influential pinning node is based on evaluating the corresponding synchronization speed. In this paper, the MASs with a single pinning node are analyzed. For MASs with multiple pinning nodes, the proposed method is still applicable. The third aim is to illustrate how the selection of the most influential pinning node can be affected by the switching frequency.

\subsection{The condition for synchronization}
For synchronization at the desired trajectory, we need to pin some nodes and exert control over them: $u_i=-w_i\Lambda \left ( x_i(t)-c(t) \right ),\quad i=1,2,...,m$,
where $w_i$ is the control gain. The error state of $i$th node is denoted by: $e_i(t)=x_i(t)-c(t),\quad i = 1,2,...,m$.

Since we have to consider the influence of switching frequency, let $\bm{e}_T(t) = [e_1(t)^T,...,e_n(t)^T]^T$ be the error state of the whole system under switching period $T$ and one has:
\begin{equation}\label{e:eT}
\begin{aligned}
	\dot{\bm{e}}_T(t) &= (I_N\otimes A-(rL(t)+W)\otimes\Lambda) \bm{e}_T(t),\\
	& \triangleq D_T(t) \bm{e}_T(t).
\end{aligned}
\end{equation}
where $W = \text{diag}(w_1,w_2,...,w_m)$.

Thus, the synchronization of the MAS can be converted to the stability problem of \eqref{e:eT}. Since the topologies are periodically switched with period $T$, one has $D_T(t+T)=D_T(t)$.

Let $R_T(t)$ be the transition matrix of system \eqref{e:eT}:
\begin{equation}\label{e:transition}
	R_T(t) =\int_{0}^{t}\exp \left(D_T(t)\right)dt.
\end{equation}
In particular, it has $R_T(T) = \coprod_{k=0}^{N-1} \exp\left(\tau D_T(t_k) \right)$.

\begin{theorem}
For MAS \eqref{eq1} under pinning control with periodically switched topologies, the synchronization state is asymptotically stable if and only if $\rho (R_T(T)) <1$.
\end{theorem}
\begin{proof}
Consider any moment $t = mT+t_1$ where $m\in\mathbb{N}$, $t_1\in[0,T)$ and the error state satisfies
\begin{equation}
	\bm{e}_T(t) = \bm{e}_T(mT+t_1)=R_T(T)^m\bm{e}_T(t_1).
\end{equation}
Then
\begin{equation}\label{e:thm1:norm_estimation}
\begin{aligned}
		\lim_{t\rightarrow \infty}\|\bm{e}_T(t)\| \leqslant \lim_{m\rightarrow \infty}e_0\exp(d_{\max}T)\|R_T(T)^m\|,\\
		\lim_{t\rightarrow \infty}\|\bm{e}_T(t)\| \geqslant \lim_{m\rightarrow \infty}e_0\exp(d_{\min}T)\|R_T(T)^m\|,\\
\end{aligned}
\end{equation}
where $d_{\max} = \max\limits_{k = 1,...,p} \|D_T(t_k)\|$, $d_{\max} = \min\limits_{k = 1,...,p} \lambda_{\min} \left(D_T(t_k)\right)$ and $e_0 \triangleq \|\bm{e}_T(0)\|$ is the norm of initial state independent from period.

Based on the Gelfand's formula \cite{schwartz1963linear}
\begin{equation}\label{e:thm1:norm_estimation2}
	\lim_{m\rightarrow \infty }\|R_T(T)^m\|^{1/m} = \rho(R_T(T)).
\end{equation}
	
Thus $\|R_T(T)^m\|$ decreases to zero exponentially if and only if $\rho(R_T(T))<1$. Obviously, ${\left \|\bm{e}_T(t) \right \|}^{2}$ converges to zero exponentially as well, which means that the synchronization state is asymptotically stable.
%
%
%
%
%
%
%
%
%
\end{proof}

\subsection{The most influential driver node}
Given limited resources, it is necessary to optimize the selection of pinning nodes in order to achieve better control effect. The definition of control effect varies, and this paper measures it by the synchronization speed of the network.

\begin{definition}\label{def:converge_speed}
	Consider any system $\dot{x}(t) = f(x(t))$ and following limitation exists and is not equal to 0:
	 \begin{equation}
	 	\lim_{t\rightarrow \infty} \dfrac{\ln \|x(t)\|}{t} = b.
	 \end{equation}
	 We call $b$ the convergence speed of the system.
	 
	 In particular, the synchronization speed of a MAS \eqref{eq1} is the  convergence speed of the error system \eqref{e:eT}.
\end{definition}

Let node $i$ be the single pinning node in the MAS. Then $D_T$ in \eqref{e:eT} becomes $D_{T,i}$ defined as
\begin{equation}
	D_{T,i}(t) = I_N\otimes A - \left ( rL(t) + W_i \right ) \otimes\Lambda
\end{equation}
where all entries of $W_i\in\mathbb{R}^{n\times n}$ are zero except the $i$th element on the diagonal which is $w_i>0$.
Denote $R_{T,i}$ the transition matrix induced from $D_{T,i}$ during a single period as \eqref{e:transition} and denote $\bm{e}_{T,i}$ corresponding error distance between agent trajectories and desired state.

Based on \eqref{e:thm1:norm_estimation2}, the synchronization speed of $\|\bm{e}_{T,i}(t)\|$ has

\begin{equation*}\label{e:synchronization_speed}
\begin{aligned}
	\lim_{t\rightarrow\infty}\dfrac{\ln (\|\bm{e}_{T,i}(t)\|)}{t}\leqslant \lim_{m\rightarrow\infty}\dfrac{\ln \left(e_0\exp(d_{\max}T)\|R_{T,i}(T)^m\|\right)}{mT}\\
	\leqslant \dfrac{\ln \left(\lim\limits_{m\rightarrow\infty}\|R_{T,i}(T)^m\|^{1/m}\right)}{T}
=\dfrac{1}{T}\ln \rho(R_{T,i}(T))\triangleq b_{T,i} 
\end{aligned}
\end{equation*}
and 
\begin{equation*}
	\begin{aligned}
		\lim_{t\rightarrow\infty}\dfrac{\ln (\|\bm{e}_{T,i}(t)\|)}{t}&\geqslant \lim_{m\rightarrow\infty}\dfrac{\ln \left(e_0\exp(d_{\min}T)\|R_{T,i}(T)^m\|\right)}{mT}\\
		&= \dfrac{1}{T}\ln \rho(R_{T,i}(T)) = b_{T,i} 
	\end{aligned}
\end{equation*}

Squeeze theorem \cite{sohrab2003basic} indicates that $b_{T,i}$ is the speed of synchronization with switched period $T$. Notice that $b_{T,i}<0$ means $\rho(R_{T,i}(T))\in[0,1)$ and the system can reach synchronization; $b_{T,i}\geqslant 0$ means $\rho(R_{T,i}(T))\geqslant 1$ which indicates the error among agents fails to converge to zero, and the system can not be synchronized with pinning node $i$.

\begin{definition}
	Let $w_i = w$, $i = 1,...,N$. We say node $i$ is the most influential node at pinning strength $w$ if $b_{T,i}\geqslant b_{T,j}$, $j = 1,...,N$, $j\ne i$.
\end{definition}

For the switching phenomenon in \eqref{eq1}, scaling the frequency may affect the synchronization speed, which can further affect the selection of the most influential node. An example is provided in simulation A to illustrate this statement.

\begin{prop}\label{thm:no_influence}
	The determination of the most influential pinning node is independent of the system matrix $A$ if the coupling matrix $\Lambda$ is an identity matrix.
\end{prop}

\begin{proof}
Proposition \ref{thm:no_influence} needs to prove $\rho(R_{T,i}(T))$, $i = 1,...,N$ proportional change with different $A$.
	
Notice that
\begin{equation*}
\begin{aligned}
	\text{exp}(D_{T,i}(t)) &= \text{exp}\left(  I_N\otimes A - \left ( rL(t)+W_i \right ) \otimes\Lambda \right)\\
	&= \text{exp}\left(  I_N\otimes A \right)\text{exp}  \left(-\left ( rL(t)+W_i \right ) \otimes\Lambda \right)\\
	&=\text{exp} \left(-\left ( rL(t)+W_i \right ) \otimes\Lambda \right) \text{exp}\left(  I_N\otimes A \right)  
\end{aligned}
\end{equation*}
Then we have 	
\begin{equation}
\begin{aligned}
	R_{T,i}(T)&=  \left[\coprod_{k = 0}^{M-1} \text{exp}\left(-(r L(t_k) + W_i)  \otimes \Lambda\right)\right]  \text{exp}\left( I_N\otimes T A \right )\\
	&=  \left[\coprod_{k = 0}^{M-1} \text{exp}\left(r L(t_k) + W_i \right)^{-1}\right]  \otimes  T A
\end{aligned}
\end{equation}
Let $\text{eig}(\cdot)$ be a vector formed by all eigenvalues of a matrix and we have:
\begin{equation}
	\text{eig}(R_{T,i}(T))=  T\text{eig}(Q_i)\otimes \text{eig}(A) 
\end{equation}
where
\begin{equation}
	Q_i = \coprod_{k = 0}^{M-1} \text{exp}\left(r L(t_k) + W_i \right)^{-1},
\end{equation}
and
\begin{equation}
	\rho(R_{T,i}(T)) = T\rho(Q_i)\rho(A).
\end{equation}
Then the influence of $A$ is applied equally to each eigenvalue of $R_{T,i}(T)$. System matrix $A$ does not affect the selection of most influential point.
\end{proof}

%
%
%

One of the most common methods to analyze the periodically switched MASs is to simplify the switching topology as a fixed average graph and establish the average system by assuming the switching frequency is fast\cite{kim2013consensus,stilwell2006sufficient}. Thus, in this research, we estimate the switching frequency threshold that can make the selected best driver nodes remain the same as the average system. Without calculating the transition matrix, it is difficult to obtain the synchronization speed accurately. But we can still propose a method to choose the best pinning node and optimize the efficiency of control.

The average matrix of the error system where only node $i$ is pinnedv is defined as:
\begin{equation}
	\bar{D}_i = I_N\otimes A - \left ( rL_{av} + W_i \right ) \otimes\Lambda.
\end{equation}

Note that $\lim\limits_{T\rightarrow 0} \dfrac{R_{T,i}(T)}{T} = \bar{D}_i$. Hence, when the period shrinks to 0, \eqref{e:eT} becomes the average system: $\dot{\bm{e}} = \bar{D}_{T,i} \bm{e}$.

Similar to definition \ref{def:converge_speed}, we can calculate the convergence speed of the average system $\bar{b}_i= \max\text{Re}\lambda(\bar{D}_i)$ which means the largest real part of eigenvalue of $D_i$. Then we can obtain the corresponding most influential node.

Consider the switching period as a varying factor and then we have following theorem:
\begin{theorem}\label{thm:3}
	If node $i_0$ is the most influential node of average system, then there exists a threshold $T_0$ such that the node remains so at $T<T_0$.
\end{theorem}
\begin{proof}

First we estimate the error between the periodically switched system and the average system:
\begin{equation}\label{e:thm2:e_Ti}
\begin{aligned}
\bm{e}_{T,i}(T) &= \coprod_{k = 0}^{M-1}\exp(D_{T,i}(k\tau)\tau)\bm{e}(0) \\
&= \coprod_{k=0}^{M-1}\left(\sum_{p=0}^\infty \frac{D_{T,i}(k\tau)^p \tau^p}{p!}\right) \bm{e}(0)
\end{aligned}
\end{equation}
and
\begin{equation}\label{e:thm2:e}
	\begin{aligned}
		\bm{e}(T) = \exp(\bar{D}_iT) \bm{e}(0) 
		= \sum_{p=0}^\infty\left(\frac{(\sum_{k=0}^{M-1}D_{T,i}(k\tau))^p \tau^p}{p!}\right) \bm{e}(0).
	\end{aligned}
\end{equation}

Let $g_{T,i}(t) \triangleq \bm{e}_{T,i}(t) - \bm{e}(T) $ be the error between two systems. Note that the sum and product of several convergent series are still convergent. Then $g_{T,i}(t)$ can be expressed as:
\begin{equation*}
\begin{aligned}
		&g_{T,i}(t) = \left(1 + \sum_{k=0}^{M-1}D_{T_i}(k\tau)\tau + O(\tau^2) \right)e(0)\\
		&-\left(1 + \sum_{k=0}^{M-1}D_{T_i}(k\tau)\tau + \sum_{p=2}^\infty\frac{(\sum_{k=0}^{M-1}D_{T,i}(k\tau))^p \tau^p}{p!} \right)e(0)
\end{aligned}
\end{equation*}
where $O(\tau^2)$ represents a polynomial of matrix with orders same or higher than $\tau^2$.

Since $d_i \geqslant \|D_{T,i}(k\tau)\|$, $k = 1,...,p$, we have 
\begin{equation*}
	\begin{aligned}
		&\|g_{T,i}(t)\|\leqslant\left( \prod_{k=0}^{M-1}\left(\sum_{p=0}^\infty \frac{d_i^p \tau^p}{p!}\right)- 1 - d_iT \right)\|\bm{e}(0)\| \\
		&~~~~~~~~~~~~~~~~~~~~~~~~~~~~~+  \left(e^{d_i T}- 1 - d_iT\right) \|\bm{e}(0)\|\\
		&=\left( \left(e^{d_i\tau}\right)^M- 1 - d_iT \right)\|\bm{e}(0)\|+  \left(e^{d_i T}- 1 - d_iT\right) \|\bm{e}(0)\|
	\end{aligned}
\end{equation*}


Then the norm of $g_{T,i}(T)$ can be estimated by 
\begin{equation}\label{e:g_Ti}
	\|g_{T,i}(T)\|\leqslant 2(e^{d_iT}-1-d_iT)\|\bm{e}(0)\|.
\end{equation}

Since $i_0$ is the most influential driver node of the average system, it has $b_{i_0}< b_j$, $\forall j\ne i_0$ and 
\begin{equation*}
\begin{aligned}
	\|e_{T,i_0}(T)\|&\leqslant \left(e^{\bar{b}_{i_0}T} + 2(e^{d_{i_0}T}-1-d_{i_0}T)\right)\|\bm{e}(0)\|\triangleq \overline{e}_{i_0}(T).
\end{aligned}
\end{equation*}
\eqref{e:g_Ti} indicates
\begin{equation*}
	\|e_{T,j}(T)\|\geqslant \left(e^{\bar{b}_jT} -2(e^{d_jT}-1-d_jT) \right)\|\bm{e}(0)\|\triangleq \underline{e}_{j}(T).
\end{equation*}

Notice that $\overline{e}_{i_0}(0) = \underline{e}_{j}(0)$, $\overline{e}_{i_0}'(0) < \underline{e}_{j}'(0)$, $\overline{e}_{i_0}''(T) > \underline{e}_{j}''(T)$, then there exists $T_0>0$ such that 
\begin{equation}\label{e:thm2:threshold}
	\overline{e}_{i_0}(T_0) = \min_{j\ne i}\underline{e}_{j}(T_0).
\end{equation}
Therefore, for any $0<T<T_0$, $\|e_{T,i_0}(T)\|<\overline{e}_{i_0}(T)<\underline{e}_{j}(T)<\|e_{T,j}(T)\|$ and $i_0$ is still the most influential node.
\end{proof}
\begin{remark}
	If the coupling matrix $\Lambda$ is an identity matrix, $\bar{b}_i$ and $d_i$ in the threshold \eqref{e:thm2:threshold} can be redefined as
	\begin{equation}
		\bar{b}_i = \rho\left(rL_{av}+W_i\right),\quad d_i = \max_{k = 1,...,p} 
		\left\|rL(k\tau) + W_i \right\|.
	\end{equation}

	It can be obtained by assuming $A$ is an empty matrix and Proposition \ref{thm:no_influence} indicates that this assumption has no effect on the calculation of synchronization speed.

\end{remark}
\section{SIMULATION EXAMPLES}
This simulation example applies the pinning control to a four-nodes system with states $x_i (t) \in \mathbb{R}^2$, $i = 1,2,3,4$. The topology of switching network is shown in Fig. \ref{fig:temporal dilation}.
\begin{figure}[htbp]
	\centering
	\includegraphics[trim = 21 27 21 22, clip, width=0.9\linewidth]{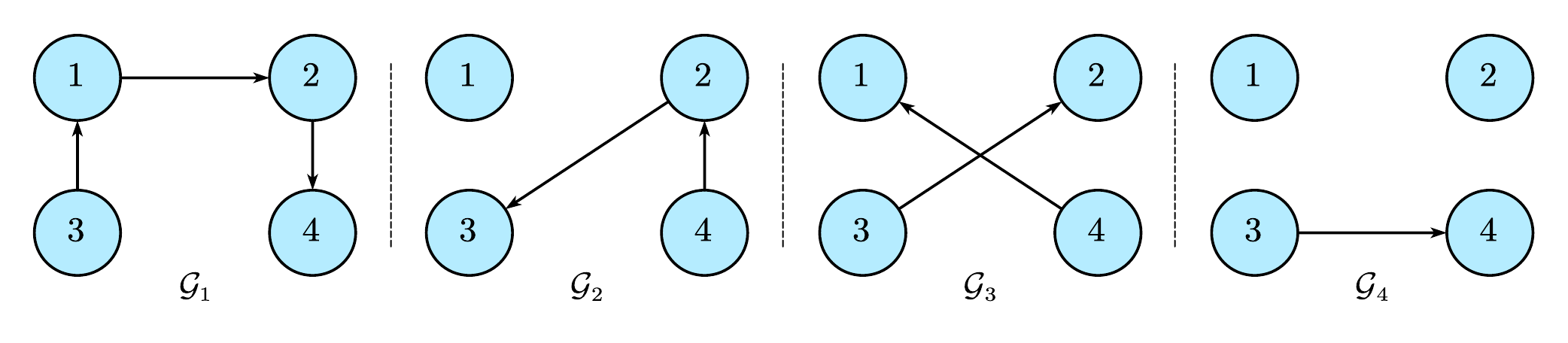}
	\caption{Switching topology of the network}
	\label{fig:temporal dilation}
\end{figure}

The switching signal is $\xi(t) = i, \quad t\in[kT+iT/4, kT+(i+1)T/4)$, where index $i = 1,2,3,4$ and $k \in\mathbb{N}$. The coupling strength $r=1$, and control gain $w_i=5$. The dynamics of desirable trajectory is $\dot{c}(t)=Ac(t)$ with initial state $c(0)=[1,1]^T$. The initial states of four nodes are $x_1(0)=[-2,5]^T$, $x_2(0)=[-2,1]^T$, $x_3(0)=[-4,1]^T$, $x_4(0)=[-3,2]^T$. In this simulation, only one node is selected to be pinned in each case.

\subsection{Synchronization}
\subsubsection{Synchronization results and speed analysis}

We first demonstrate the synchronization of a MAS with periodically switching topologies based on Theorem 1. In system 1, the inner coupling matrix $\Lambda=I_2$, and the system parameter is denoted as $A_1$,
\begin{equation*}
	A_1 = \jz{0.320& 0.304\\-0.544&0.112}.
\end{equation*}

The relationship between synchronization speed $b_i$ and switching period $T$ of system 1 is shown in Fig. \ref{fig:synchronization_speed_system1}. It shows that when switching period $T = 6$ (shown by the dashed line), $b_1>0$ and $b_2<b_4<b_3<0$. Based on Theorem 1, this result indicates that system 1 can not achieve synchronization when node 1 is pinned, but it can be synchronized in the other three pinning node cases. The results in Fig. \ref{fig:synchronization error} confirm this conclusion. In Fig. \ref{fig:synchronization error}, the error states between the desirable trajectory and the system states are provided. It can be observed from the first case of Fig. \ref{fig:synchronization error} that the error states diverge with time when node 1 is pinned, which indicates the system can not be synchronized. In the other three cases, the error states converge to zero when node 2, 3 and 4 are pinned, which means the system achieves synchronization.

The synchronization speed can be revealed from Fig. \ref{fig:synchronization error}. The converging speed of error states is the fastest when node 2 is pinned. The second fastest case is choosing node 4 as the pinning node. The slowest one is the case with pinning node 3. This phenomenon is consistent with the results $b_2<b_4<b_3<0$ shown before from Fig. \ref{fig:synchronization_speed_system1}. To obtain the fastest synchronization speed, the case with smallest $b_i$ is identified as the most influential driver node. Hence, node 2 is the best driver node for system 1 when switching period $T = 6$.
\begin{figure}[htbp]
	\centering
	\includegraphics[trim = 100 420 100 280, clip, width=\linewidth]{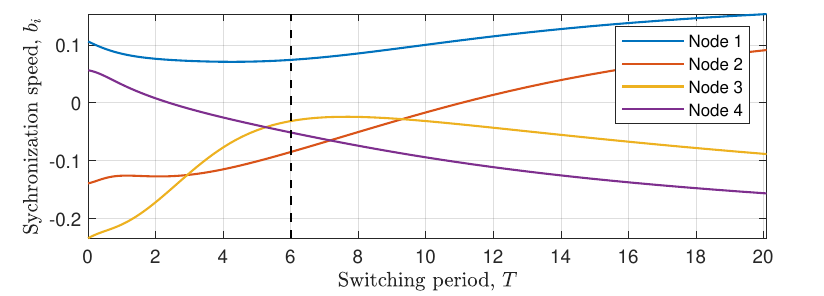}
	\caption{Synchronization speed of system 1}
	\label{fig:synchronization_speed_system1}
\end{figure}
\begin{figure}[htbp]
	\centering
	\includegraphics[trim = 100 270 100 270, clip, width=\linewidth]{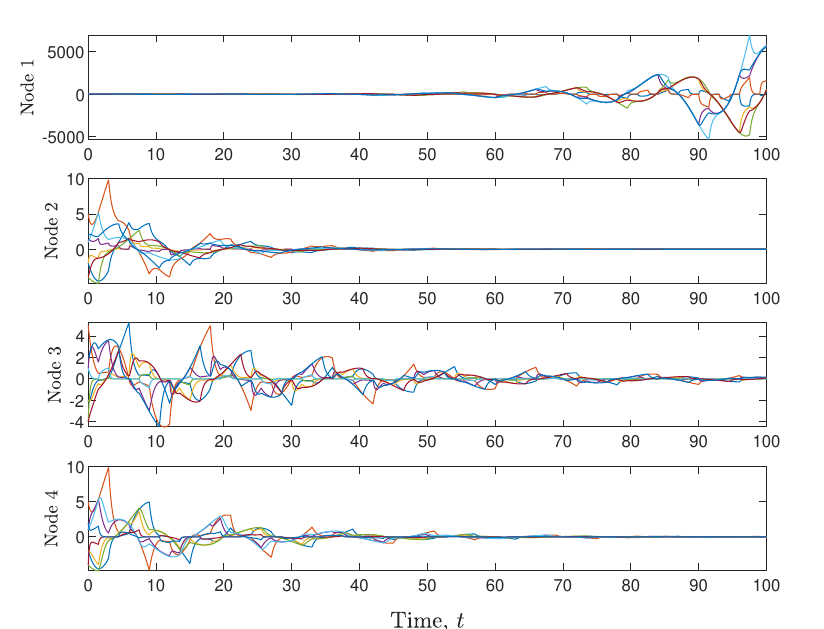}
	\caption{Synchronization error}
	\label{fig:synchronization error}
\end{figure}

\subsubsection{Independence of system matrix $A$}
Simulation of system 2 is generated with the same parameter settings as system 1 except system matrix $A$ defined in \eqref{eq1}. The system matrix of system 2 is denoted as $A_2$,
\begin{equation*}
	A_2 = \jz{0.500&0.400\\-0.200&0.300}.
\end{equation*}

This simulation shows the synchronization speed $b_i$ with varying switching period $T$ of system 2 in Fig. \ref{fig:synchronization_speed_system2}. By comparing the results between these two systems, it can be found that the trajectories of four pinning node cases are only scaled  from Fig. \ref{fig:synchronization_speed_system1} to Fig. \ref{fig:synchronization_speed_system2}, and the relative magnitude relationship between them remains consistent throughout. Hence, changing system parameter does not affect the identification of the most influential driver node, which verifies the Proposition 1.

\begin{figure}[htbp]
	\centering
	\includegraphics[trim = 100 420 100 280, clip, width=\linewidth]{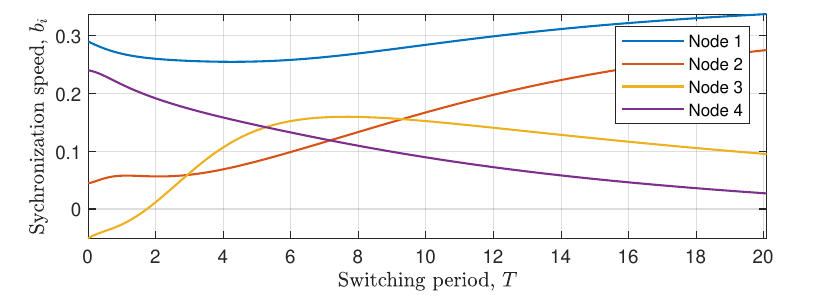}
	\caption{Synchronization speed of system 2}
	\label{fig:synchronization_speed_system2}
\end{figure}

\subsection{Switching period threshold $T_0$}
In this simulation example, the system parameter is set to be $A=A_1$, and the inner coupling matrix is
\begin{equation*}
	\Lambda =\jz{1.6&-1.0\\-0.6&0.8}.
\end{equation*}

As is shown in Fig. \ref{fig:threshold}, pinning node 3 leads to the fastest synchronization speed in average system when switching period $T=0$. Before the first bifurcation point $T=1.57$, the most influential pinning node remains as node 3.
\begin{figure}[htbp]
	\centering
	\includegraphics[trim = 100 420 100 280, clip, width=\linewidth]{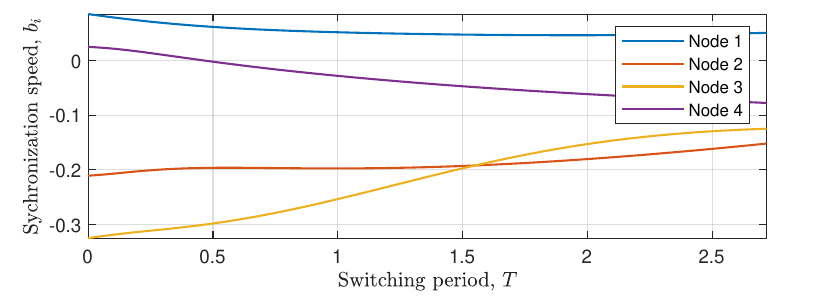}
	\caption{Real switching period threshold}
	\label{fig:threshold}
\end{figure}

According to Theorem \ref{thm:3}, the estimated threshold satisfies for $j = 1,2,4$:
\begin{equation}
	e^{\bar{b}_3T_0} + 2(e^{d_3T_0}-1-d_3T_0)\leqslant e^{\bar{b}_jT} -2(e^{d_jT}-1-d_jT)
\end{equation}

%

By numerical calculation, the estimated threshold is located at $T_0=4.5\times 10^{-4}$, which is smaller than the real bifurcation point $T=1.57$ and  verifies Theorem \ref{thm:3} therefore.

%

\section{CONCLUSION}

Pinning control of MASs has many potential applications in various engineering fields. One of the key challenges is identifying the most influential driver nodes of pinning control, especially in the systems with the time-varying topology. By analyzing the state transition matrix, we have provided the pinning synchronization conditions of MASs with periodically switched topologies. A method has been proposed to find the most influential driver nodes that can guarantee the fastest synchronization speed. This paper also analyzes the impact factors of identifying the best driver nodes, including the topology switching frequency and system parameter matrix. Theoretical proof and simulation have indicated that the determination of the most influential driver nodes is independent of the system parameter matrix when the inner coupling matrix is an identity matrix. Furthermore, an estimation method has been developed to find the switching frequency threshold that can keep the best driver node identified in the average system unchanged.

\bibliographystyle{IEEEtran}
\bibliography{Mybib}

\end{document}